\documentclass{article}[11pt]
\usepackage[utf8]{inputenc}
\usepackage{amsmath,amsthm}
\usepackage{amssymb}
\usepackage{xcolor}
\usepackage[capitalise]{cleveref}
\usepackage{graphicx}
\usepackage{multirow}

\newcommand{\ignore}[1]{}
\newcommand{\R}{\mathbb{R}}

\newcommand{\Set}{\mathcal{C}}

\newcommand{\comment}[1]{\textcolor{red}{\textbf{#1}}}
\newcommand{\shlomo}[1]{\comment{\textcolor{blue}{Shlomo: #1}}}

\newtheorem{thm}{Theorem}

\newtheorem{lem}[thm]{Lemma}
\newtheorem{defi}{Definition}

\makeatletter
\let\xx@thm\@thm
\AtBeginDocument{\let\@thm\xx@thm}
\makeatother

\makeatletter
\let\xx@prop\@prop
\AtBeginDocument{\let\@prop\xx@prop}
\makeatother

\makeatletter
\let\xx@cor\@cor
\AtBeginDocument{\let\@cor\xx@cor}
\makeatother

\title{Adjustable Coins}

\author{Shlomo Moran\thanks{Compter Science Department, Technion. Israel}
	\and Irad Yavneh\footnotemark[1]
 }
\date{\today}

\begin{document}

\maketitle

\noindent
\begin{abstract} 
	 {
	  In this paper we consider a scenario  where there are several algorithms for solving a given problem. Each algorithm is associated with a probability of success and a cost, and there is also a penalty for failing to solve the problem. The user may run one algorithm at a time for the specified cost, or give up and pay the penalty. The probability of success may be implied by randomization in the algorithm, or by assuming a probability distribution on the input space, which lead to different variants of the problem. The goal is to minimize the expected cost of the process under the assumption that the algorithms are independent.  We study several variants of this problem, and present possible solution strategies and a hardness result.
}
  \end{abstract}
\section{Introduction}
Some optimization problems concern optimal ordering  of certain tasks that aim at achieving some common goal (see, e.g., \cite{BBCSZ14}). A possible scenario can be described as follows:
 the king’s daughter is approaching marrying age. The miser king, who cares about money more than anything else, estimates that if she doesn’t marry then he will need to spend $E$ gold pieces to continue supporting her for the rest of their lives. There are $n$ princes he can invite to try to win her heart, and he must choose the order, knowing that the cost of travel room and board for $prince_i$ is $\mu_i$, and the probability that he will win the princess’s heart is $P_i$. What order minimizes the king's expected cost?  In a more challenging scenario, there are $n$ kingdoms with several eligible princes in each kingdom, each with his own $\mu$ and $P$, and the king must also decide which prince he should invite from each kingdom (he cannot invite more than one).


Such problems can be described as one player games, which we call {\em adjustable coins} games.
%
	%
	An adjustable coin---A-coin in short---is a coin whose bias can be controlled by the user: the probability of success (rolling on one) can be increased, for a price. Formally, an A-coin is defined by a monotone increasing function $\mu:D_\mu\rightarrow \R^+$ where $D_\mu\subseteq (0,1]$ contains the possible success probabilities, and for $P\in D_\mu$, $\mu(P)$ is the nonnegative fee for tossing the coin with   success probability $P$. If $|D_\mu|=1$ then $\mu$ is a {\em simple  coin}, or just coin,    denoted by a pair $c=(P,\mu)$. Thus an   A-coin $\mu$ can be viewed as a set of the simple coins $\{(P,\mu(P)):P\in D_\mu\}$.
	
	In the games studied in this paper the player is given a set $\Set$ of A-coins and a penalty $E>0$, and in each step she may select an A-coin $\mu\in \Set$, and  a probability $P\in D_\mu$, and toss the coin for the fee $\mu(P)$ . If the coin rolls on one the player  terminates the game without paying a penalty, else she either takes another step, or terminates the game and pays the penalty $E$. Thus, a strategy for an A-coin game maps each pair $(\Set,E)$ of the set of A-coins and a penalty to a sequence $SEQ$ of the coins tossed by the player when the outcomes of all tosses are zeros. The goal is to minimize the expected cost---the total amount paid. Variants of this game are determined by the nature of the coins in $\Set$, the rules by which coins can be selected at each step, possible restrictions on the termination rule, etc. Specifically, we distinguish between reusable coins, which can be tossed many times, and one time coins, which can be tossed only once. The latter case corresponds to a scenario where one or more deterministic tests should  be taken on a single item selected at random from a known distribution---repeating a test on the selected item just reproduces the initial outcome.
	
	In Section \ref{sec:basic} we study variants of the game for simple coins.  In Section \ref{sec:discrete} we   study  the case where the A-coins are discrete, i.e., defined for finitely many values. In Section \ref{sec:continuous} we study the game for A-coins which are  piecewise continuous functions. 

 \section{Simple coins}\label{sec:basic}
In this section we study optimal strategies for the game where all A-coins are simple. A useful property of a simple coin $c=(P,\mu)$ is its {\em rate}, given by the ratio $r=\mu/P$. The notation $c\sim (P,r)$ means that the simple coin $c$ has probability $P$ and is of rate $r$, i.e., $c=(P,rP)$.
 \subsection{Single simple coin, single toss} \label{sec:SinglecoinSingleToss}

 In the simplest scenario  we are given a simple coin $c=(P,\mu)$ and a penalty $E$, and we must decide whether tossing the coin $c$ is {\em beneficial}, i.e., reduces the expected cost of the game.

  The expected payment when $c$ is tossed is 	 given by
 \begin{equation}\label{eq:mu}
 COST\left(c, E \right)  = \mu + (1-P)E   = E -(PE-\mu),
 \end{equation}
 which means that the {\em benefit} of using $c$ w.r.t. $E$ is $(PE-\mu)$ (see Figure \ref{fig:benefit}).
  $COST(c,E)$ as a function of  the rate $r$ of $c$   is given by
 \begin{equation}\label{eq:qu}
 COST\left(c, E \right)  = P  r + (1-P)E	 = E-P(E-r)  ,
 \end{equation}
 implying that the benefit of $c$ w.r.t. $E$ is $P(E-r)$. This means that  an optimal strategy for this case is to toss $c$ if and only if its rate is smaller than $E$.
  \begin{figure}[t!]
 	\begin{center}
 		 \vspace{-180pt}
 		\includegraphics[width= \columnwidth]{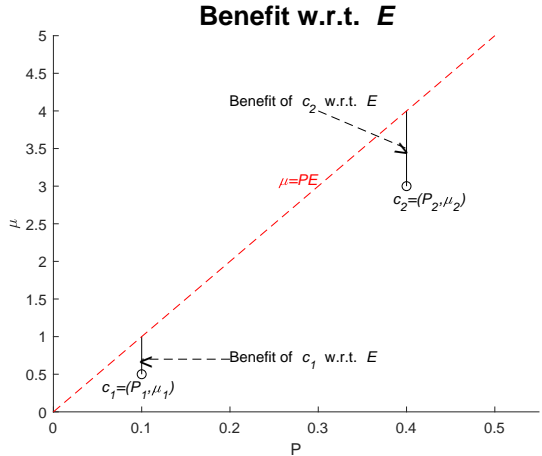}
  	\vspace{-150pt}
 		\caption{\textbf{Benefit for a given penalty.} {A coin $c=(P,\mu)$ is beneficial for penalty $E$ iff $(P,\mu)$ lies below the line $\mu=PE$, and the benefit of $c$ w.r.t. $E$ is $PE-\mu$. Thus both $c_1$ and $c_2$ are beneficial for $E$. The benefit of $c_2$ is larger since it is further away from the line $\mu=PE$, i.e. $\frac{\mu_2-\mu_1}{P_2-P_1}<E$.}
 		}
 		\label{fig:benefit}
 	\end{center}
 \end{figure}

 In a natural extension of the ``single coin single toss'' game we are given a sequence of coins $SEQ=(c_1,\ldots,c_n)$ and a penalty $E$. Our task is to decide for each coin, in the given order, if it should be tossed or skipped, so that the expected cost of the implied game is minimized. For $n=1$ this is the single coin single toss game.  {For   $n>1$ we use {\em backwards induction}:} Suppose
  we have an optimal strategy for $SEQ'=(c_2,\ldots,c_n)$ and $E$, whose cost is $E'$. Then, by a simple calculation,  
   an optimal strategy for   $SEQ$ and $E$   is: toss $c_1$ iff $r_1<E'$, and then,  {if the coin rolls on tail,} continue with the optimal strategy for $SEQ'$ and $E$.

 \subsection{Multiple simple coins, single toss} \label{sec:MultipleCoinSingleToss}
 Suppose next that we are given a {(finite)} set $\Set$ of  simple coins and a penalty $E$,   and  we wish to select an {\em optimal} coin $c\in\Set$ which minimizes $COST(c,E)$, the cost of the single toss game. Let $c=(P,\mu)$.
   By Equation (\ref{eq:mu}), $COST(c,E)$ is minimized  if and only if 
  the point $c=(P,\mu)$   lies below the line $\mu=PE$ and at a maximal distance (see \cref{fig:benefit}).  



 \subsection{Reusable simple coins, multiple tosses} \label{sec:ReusableSimpleCoins}
Suppose now that the coins in $\Set$ are {\em reusable}, that is, we may toss each coin in $\Set$   multiple times.  An elementary calculation shows that if there is no bound on the number of tosses, then tossing a coin of minimal possible rate $r_{min}$ until it rolls on one yields an expected cost $r_{min}$.  {Thus an optimal strategy is: if $r_{min}<E$ then toss a coin $c_{min}$ of rate $r_{min}$ until it rolls on one, else pay the penalty $E$ and terminate.} 

 Assume now that we are allowed to make no more than $k$ tosses for some   $k\geq 0$. Given a set of reusable simple coins $\Set$ and a penalty $E$,  {if $r_{min}\geq E$ then do nothing and pay the penalty $E$. Else an optimal strategy is obtained again by a backwards induction:} 

 {If $k=0$ or $E\leq r_{min}$ then do nothing and pay the penalty $E$. 
	So assume that $E > r_{min}$, and 
	 let   $E_k=E_k(\Set,E,k)$ be the expected cost of an optimal strategy for set of coins $\Set$, penalty $E$ and $k$ tosses (in particular $E_0=E$). Then $E_k > r_{min}$ and an optimal strategy for $k+1$ tosses  is obtained 
   by  first selecting and tossing a coin $c_{k+1}\in\Set$ which minimizes $COST(c_{k+1},E_k)$  as in Section \ref{sec:MultipleCoinSingleToss} (note that $COST(c_{k+1},E_k)<E_k$ since  $E_k > r_{min}$);
  if $c_{k+1}$ rolls on one then stop, else execute the optimal strategy for $\Set,E$ and $k$. Thus $E_{k+1}=COST(c_{k+1},E_k)>r_{min}$.}

 \subsection{One-time simple coins, multiple tosses} \label{sec:NonreusableSimpleCoins}
 We now assume that each coin in $\Set$ can be tossed at most once, and we are allowed to toss as many coins as we wish.

 Consider first a variant of this game in which termination is possible only if either  some coin rolls on one, or  after {\em all} coins rolled on zero. We note that  this variant is, in fact, the problem of  optimal ordering of independent tests that was studied in \cite{BBCSZ14}, where each coin corresponds to a test, and it is needed to check if at least one test fails.

 A strategy for this latter game is a permutation of all the available coins. Thus given a set of $n$ one-time simple coins $\Set=\{c_1,\ldots,c_n\}$, where $c_i=(P_i,\mu_i)$, and a nonnegative penalty $E$,
 we  need to find an {\em optimal ordering} of the  coins in $\Set$, which  minimizes the expected cost. 

 The expected cost  for  ordering  $SEQ=(c_1, c_2,..., c_n)$ and penalty $E$ 
 is given by:
 \begingroup\small
 \begin{equation}\label{eq:mus}
 \begin{array}{lll}
  COST(SEQ,E)& =&{  P_1 \mu_1  + (1-P_1)P_2(\mu_1+\mu_2)+ \ldots}\\
 &+& { (1-P_1)(1-P_2)\cdots(1-P_n)(\mu_1+\cdots+\mu_n+E)}.
 \end{array}
 \end{equation}
 \endgroup
 By straightforward induction, $COST(SEQ,E)$ can also be expressed as the following convex combination of the coin rates $r_1,\ldots,r_n$  and  $E$.
  \begingroup\small
 \begin{equation}\label{eq:qs}
 \begin{array}{lll}
COST (SEQ,E)   &=&  P_1r_1 +   (1-P_1)P_2r_2 +\ldots +    (1-P_1)\cdots(1-P_{i-1})P_ir_i      \\
   &+& \ldots ~~+ (1-P_1)\cdots (1-P_n)E   .
 \end{array}
 \end{equation}
  \endgroup
Equation (\ref{eq:qs}) implies the following useful lemma:
\begin{lem}\label{lem:costreduce}
	Let  $SEQ=(c_1,\ldots,c_n)$, where $c_k\sim(P_k,r_k),k=1,\ldots,n$, and let  $SEQ'$ be obtained from $SEQ$ by interchanging $c_i$ and $c_{i+1}$.  If $r_i\geq r_{i+1}$ then $COST(SEQ',E)\leq COST(SEQ,E)$, with equality iff $r_i=r_{i+1}$.
\end{lem}
\begin{proof}
	By substituting in Equation (\ref{eq:qs}) we get
	$$
	COST(SEQ',E) - COST(SEQ,E) =\bigg(\prod_{k=1}^{i-1}(1-P_k)\bigg)P_iP_{i+1}(r_{i+1}-r_{i}) ,
	$$
	which is negative for $r_i>r_{i+1}$, and equals zero iff $r_i=r_{i+1}$.
\end{proof}
Lemma \ref{lem:costreduce} implies: 
 \begin{lem}\label{lem:sequence}
 	Let $\Set=\{c_1,\ldots,c_n\}$ be a set of {one-time} simple coins, where the rate of $c_i$ is $r_i$. Then for each penalty $E$ and each permutation $\pi$, $(c_{\pi(1)},\ldots,c_{\pi(n)})$    is an optimal ordering of $\Set$  w.r.t $E$  if and only if  $r_{\pi(i)}\leq r_{\pi(i+1)}$ for $i=1,\ldots,n-1$.
 \end{lem}
\noindent
Observe that   the optimal orderings of a set of coins are independent of the value of the penalty $E$. We note that Lemma \ref{lem:sequence} is equivalent to Theorem 1 of \cite{BBCSZ14} which considered optimal ordering of independent tests.

 Assume now that the player can terminate the game at any time (i.e., even if no coin rolled on one and some coins   were not tossed yet). By Lemma \ref{lem:sequence} and  the comment at the end of Section \ref{sec:SinglecoinSingleToss}, an optimal strategy is obtained by an optimal ordering of the coins in $\Set$ whose rates are smaller than $E$. 
 \begin{lem}
 	\label{lem:simple}
 	The optimal strategies for a set $\Set$ of one-time simple coins and a penalty $E$ are obtained by the optimal orderings of the coins in $\Set$ whose rates are smaller than  $E$.
 \end{lem}

 \subsubsection{One-time simple coins, bounded number of tosses}
 Assume now that the coins are not reusable, and we may use at most $k$ coins for some   $k\geq 0$. This problem can be solved  by the following dynamic programming algorithm.
 First sort the coins whose rates are smaller than $E$ by increasing rates (ties are broken arbitrarily). Let the sorted list be $(c_1,\ldots,c_n)$, where $c_i=(P_i,\mu_i)$.\\
 For $i=1,\ldots,n$ and $j=0,\ldots,\max(k,n-i+1)$, let $OPT(i,j) $ be the value of the optimal strategies for the sequence $(c_i,c_{i+1},\ldots,c_n)$ and penalty $E$ which use  {at most}  $j$ coins. Thus our task is to find $OPT(1,k)$. This can be done in $kn$ steps by setting {$OPT(i,0)=E$ for $i=1, \ldots,n$, $OPT(n,j)=\mu_{n}+(1-P_n)E$ for $j\geq 1$, }and then using  the following recursive formula for $i=n-1,n-2,\ldots,1,j=1,2\ldots k$:
 $$
 OPT(i,j)=
 \min\bigg(   OPT(i+1,j), ~~~
 \mu_i+(1-P_i)OPT(i+1,j-1)~~\bigg).
 $$
 This implies the following.
 \begin{lem}\label{lem:bounded}
 	Given a set $\Set$ of $n$ one-time simple coins and bound $k$ on the number of tosses, an optimal strategy for the implied game can be found in $O(kn)$ time.
 \end{lem}

 {The results for simple coins are summarized in \cref{tab:simple}.}
 \vspace{-5pt}
\begin{table}[ht]
	\caption{Optimal solutions for Simple Coins}
	\vspace{-15pt}
	\begin{center}
		\begin{tabular}{|c|c|c|c|c|}
			\hline
			\#coins&Reuse&Tosses&Discussed&Comments\\	
			\hline\hline
			1&No&1&\multirow{3}{*}{Section 2.1}&Toss if $rate < E$\\
			\cline{1-3}\cline{5-5}
			1&Yes&Unlimited& &  $rate<E\rightarrow$ toss until success \\
			\cline{1-3}\cline{5-5}
			$n$&No& Unlimited, order given& &Backwards induction\\
			\hline
			$n$&Yes&Unlimited&\multirow{2}{*}{Section 2.3}&Toss $c_{min}$ until success\\
			\cline{1-3}\cline{5-5}
			$n$&Yes& Bounded &&Backwards induction\\
			\hline
			$n$&No& Unlimited &Lemma 3&Toss by increasing rates\\
			\hline
			$n$&No& Bounded &Lemma 4&Backwards dynamic programming\\
			\hline
		\end{tabular}
	\end{center}
	\label{tab:simple}
\end{table}
 \section{Discrete coins}\label{sec:discrete}

  \begin{figure}[h!]
 	\begin{center}
 		\vspace{-150pt}
 		 \includegraphics[width= \columnwidth]{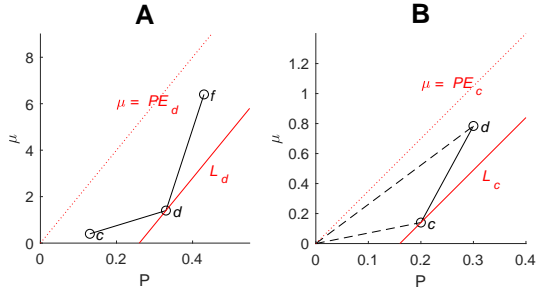}
 		\vspace{-180pt}
 		\caption{
 			 {A:  $L_d$, the supporting line for $d$ and $E_d$,   is a supporting line of the convex hull of $\mu=\{c,d,f\}$. B: The rate of  $c$ (i.e., the slope  of the dashed line segment connecting $c$ to the origin) is smaller than the rate of $d$,   (the slope of the line connecting $d$ to the origin).
 			} 	 		}
 	 		\label{fig:convex}
 	\end{center}
 \end{figure}

 An A-coin $\mu$ is discrete if its   domain $D_\mu$ is finite and contains at least two simple coins.
  Given a penalty $E>0$, $COST(\mu,E)$ is naturally defined as
 	$$COST(\mu,E) = \min_{c\in\mu}COST(c,E).$$
An A-coin $\mu$ can be viewed as the set of the simple  coins $\{(P,\mu(P)):P\in D_\mu\}$. We assume that this set does not contain redundant coins, in the sense of the following definitions.
 	\begin{defi}\label{def:essential}
 	A coin   $c\in\mu$ is {\em essential} for $\mu$ (or just essential when $\mu$ is clear) if there is a penalty $E_c$ s.t. for any other coin $d\in\mu$, $COST(c,E_c)<COST(d,E_c)$. Given  such $c$ and $E_c$, the {\em supporting line for $c$ and $E_c$} is the line $L_c$ which contains   $c$ and is parallel to the line $\mu=PE_c$ (see \cref{fig:convex}).
 	\end{defi}
 Note that if $c$ is not essential then for each $E>0$,   $COST(\mu\setminus \{c \},E)=COST(\mu,E)$, meaning that $c$ can be removed from $\mu$ without reducing its quality. Hence, we assume from now on that  a discrete  coin contains only essential simple coins.
 \begin{defi}\label{def:effucuent} A discrete coin $\mu$ is {\em efficient} if each coin in $\mu$ is essential.
 \end{defi}

Efficient discrete coins posses nice geometrical properties, depicted in \cref{fig:convex}:
 	\begin{lem}\label{lem:convex}
 		Let $\mu$ be an efficient discrete coin. Then
 		\begin{enumerate}
 			\item\label{item:convex1}  $\mu$ is a strictly convex function on $D_\mu$, and
 			\item\label{item:convex2} The function $r(P)=\mu(P)/P$ is strictly increasing on $D_\mu$.
 		\end{enumerate}
 	\end{lem}
 \begin{proof}[Sketch of proof.]
 	 Let $c\sim (P_c,r_c)$ be any coin in $\mu$.   Then $L_c$,  the supporting line for  $c$ and $E_c$, is  a supporting line of the convex hull of $\mu$ which contains $c$ but no other coin in $\mu$  (see \cref{fig:convex}A). This proves (\ref{item:convex1}).\\
 	 To show (\ref{item:convex2}), let $d\sim (P_d,r_d)$ be another coin in $\mu$,  where $P_c<P_d$.
 	Then $r_c<E_c$ (since $c$ lies below the line $\mu=PE_c$), and $d$ lies strictly above $L_c$, the supporting line for $c$ and $E_c$,  whose slope is $E_c$.  Hence $r_c<r_d$   (see \cref{fig:convex}B).
  	 \end{proof}

\subsection{Reusable discrete coins}
Optimal strategies for reusable  discrete coins are   implied by such strategies for sets of simple coins:
 Given a set $\Set$ of discrete coins and a penalty $E$, let $\Set'$ be the  union of the discrete coins in $\Set$. Applying   the  optimal strategies for reusable simple coins presented in Section \ref{sec:ReusableSimpleCoins} on $\Set'$ and $E$ yields optimal strategies for $\Set$ and $E$. For example, when the number of tosses is unbounded, \cref{lem:convex} implies the following  strategy for a reusable efficient discrete coin $\mu$: Let $c_{min}\sim(P_{min},r_{min})$ be the coin with the minimum success probability in $\mu$. If $E>r_{min}$ then repeatedly toss $c_{min}$ until it rolls on one, else pay the penalty $E$ and terminate.


 \subsection{One-time discrete coins}\label{sec:norep}
  {Suppose we are given a sequence of one-time discrete coins and we need to decide for each coin $\mu$, in its turn, whether to toss a simple coin from $\mu$, and if so which one. Then we can use a simple extension of the strategy for simple coins in \cref{sec:SinglecoinSingleToss}: Given an optimal strategy for a sequence $(\mu_2,\ldots,\mu_n)$ with optimal cost $E'$, an optimal strategy for $(\mu_1,\mu_2,\ldots,\mu_n)$ is obtained by skipping $\mu_1$ if $COST(\mu_1,E')\geq E'$, and tossing a coin $c\in\mu_1$ which minimizes $COST(\mu_1,E')$ and upon failure executing the optimal strategy for $(\mu_2,\ldots,\mu_n)$ otherwise.}

 Finding an optimal strategy for one-time  discrete coins  {when the sequence is not given} is complicated by the fact that we need to decide which simple coin should be selected from each discrete coin (if any) {before the order of tosses is known}.  Once these coins are selected, 
 all we need to do is to toss them in an increasing order of their rates, until some coin rolls on one (or all coins roll on zero), as discussed in Section \ref{sec:NonreusableSimpleCoins}. Thus, the strategy is determined by the way the simple coins are chosen from the given discrete ones.
 We next show  that such a selection 
 is NP Hard even in a highly restricted variant of the problem. 


 \subsubsection{ The  $\{0,1\}$ discrete coins problem}

 We now present a restricted version of the A-coins problem---the $\{0,1\}$ A-coins problem, and prove that it is  NP hard.  An  instance $(\Set,E)$ for this problem consists of a set  $\Set=\{A_1,\ldots,A_n\}$ {of one-time discrete coins}, where each  $A_i$  contains two coins $c_i, d_i$, s.t. the rate of all $c_i$ coins is $ 0$ and the rate of all $d_i$ coins is $1$, i.e., $c_i=(0,P_{i,0})$ and $d_i=(P_{i,1},P_{i,1})$, where $0\leq P_{i,0}<P_{i,1}\leq 1, i=1,\ldots,n$.
 Let $h_i=1-P_{i,0}$ and $\ell_i=1-P_{i,1}$; then $h=(h_1,\ldots,h_n)$ and $\ell=(\ell_1,\ldots,\ell_n)$  are the  {\em failure probability vectors} of $c_i$'s and $d_i$'s (note that $h_i>\ell_i$ for all $i$).

 \subsubsection*{NP Hardness of the $\{0,1\}$ discrete coin  problem}
 Consider an instance to the $\{0,1\}$ problem with penalty $E>1$, and let $D=E-1$. By Lemma \ref{lem:simple} an optimal strategy for this instance  is obtained by selecting a set  $S\subseteq I_n$ of indices $i$ for which the coin $c_i$ is chosen, and tossing the coins in $S$ first.  The cost of this strategy can be calculated as follows: charge each event (sequence of tosses) in which the first $|S|$ tosses are zero by one, and in addition charge the event  in which all tosses are zero by $D$. The resulting cost  is
 \begin{equation}\label{eq:cost1}
 COST_{h,\ell,D}(S) = \prod_{i\in S}h_i\left(1+D\cdot\prod_{i\in I_n\setminus S}\ell_i\right).
 \end{equation}
 For a vector $a=(a_1,\ldots,a_n)$, let $prod(a)$ be the product $\prod_{i=1}^n a_i$.
 Then
 $\prod_{[i\in  I_n\setminus S]}\ell_i = \frac{prod(\ell)}{\prod_{i\in S}\ell_i}$.
 Hence, letting $b_i=\frac{\ell_i}{h_i}<1$ and $b=(b_1,\ldots,b_n)$, we can rewrite (\ref{eq:cost1}) as a function of $h$ and $b$:

 \begin{equation}\label{eq:cost2}
 COST_{h,b,D}(S) = \prod_{i\in S}h_i +\frac{D\cdot prod(\ell)}{\prod_{i\in S}b_i}
 \end{equation}
 Let $H_S=\prod_{i\in S}h_i$ and $B_S= \prod_{i\in S}b_i$. Then Equation (\ref{eq:cost2}) can be rewritten as

 \begin{equation}\label{eq:cost3}
 COST_{\ell,b,C}(S) = H_S +\frac{C}{B_S},
 \end{equation}
 where $C=D\cdot prod(\ell)$.
 Consider now the case $h=b$ (i.e., $\forall i:\ell_i=h_i^2$).
 Then we get
 \begin{equation}\label{eq:cost3oneD}
 COST_{h,C}(S) = H_S +\frac{C}{H_S}.
 \end{equation}
 Since the function $f(x)=x+\frac{C}{x}$ has a unique minimum at $x=\sqrt{C}$, we get that $COST_{h,C}(S)\geq \sqrt{C}+\frac{1}{\sqrt{C}}$, with equality iff $H_s=\sqrt{C}$. This implies the following:
 \begin{lem}\label{lem:H_S}
 	Let $\Set=\{A_i:i=1,\ldots,n\}$ be a set of A-coins with $A_i=\{c_i,d_i\}$ where
 	\[ c_i \sim(0,1-h_i),~~d_i\sim (1,1-h_i^2) \]
 	and let $C=D\cdot\left( prod(h)\right)^2$.
 	Let further $OPT(\Set,D+1)$ be the value of the optimal solution to the   $\{0,1\}$ A-coins problem for $\Set$ and penalty $D+1$. Then
 	\[OPT(\Set,D+1)\geq \sqrt{C}+\frac{1}{\sqrt{C}},\]
 	with equality iff for some $S\subseteq I_n$ it holds that $H_S=\sqrt{C}$.
 \end{lem}
 Lemma \ref{lem:H_S} now implies:

 \begin{thm}\label{thm:NPH1}
 	The $\{0,1\}$ A-coins problem is NP Hard.
 \end{thm}
 \begin{proof}[Outline of proof]
 	By a reduction from the NP hard {\em subset product} problem \cite{M81}:\\
 	{\bf Input:} An $n+1$ tuple of natural numbers $(m_1,\ldots,m_n,N)$\\
 	{\bf Property:}   There is a subset $S\subseteq I_n$ s.t. $\prod_{i\in S}m_i = N$.
 	
 	Given an instance  $(m_1,\ldots,m_n,N)$ to the subset product problem, we reduce it to an instance $(\Set,D+1)$ to the $\{0,1\}$ problem in which $\Set=\{A_1,\ldots,A_n\}$, where for each $i$, $c_i\sim (0,1-\frac{1}{m_i})$ and  $d_i\sim (1,1-\frac{1}{m_i^2})$ - i.e., $h_i=\frac{1}{m_i}$ and $\ell_i =\frac{1}{m_i^2}$. In addition, we set $D$ to $\left(\frac{prod(m)}{N} \right)^2$.
 	
 	Let $C=D\cdot \left(prod(h)\right)^2=\frac{D}{\left(prod(m)\right)^2}$. Then from Lemma \ref{lem:H_S} it follows that  $OPT(\Set,D+1)\geq \sqrt{C}+\frac{1}{\sqrt{C}}$,  and $OPT(\Set,D+1) = \sqrt{C}+\frac{1}{\sqrt{C}}$ iff there is a subset $S \subseteq I_n$ s.t. $\prod_{i\in   S}m_i=N$. The theorem follows.
 \end{proof}
 \noindent {\bf Note:} The rates 0 and 1 in  \cref{thm:NPH1} could be replaced by any pair of distinct rates $0\leq a<b$.

A summary of the results for discrete coins appears in \cref{tab:discrete} below.
\vspace{-10pt}
	\begin{table}[ht]
	\caption{Solutions and hardenss results for Discrete Coins}
	\vspace{-10pt}
	\begin{center}
		\begin{tabular}{|c|c|c|c|c|}
			\hline
			\#Coins&Reuse&Tosses&Discussed&Comments\\
			\hline\hline
			1&No&1&Lemma 5&\\
			\cline{1-4}
			$n$&Yes&Unlimited  &\multirow{2}{*}{Section 3.1}&Extensions of solutions\\
			\cline{1-3}
			$n$&Yes&Bounded&&{for simple coins}\\
			\cline{1-4}
			$n$&No&Unlimited, order given&Section 3.2& \\
			\hline
			$n$&No&Unlimited &Theorem 7&{NP hard even for \{0,1\} coins}\\
			\hline
		\end{tabular}
	\end{center}
	\label{tab:discrete}
\end{table}
\vspace{-20pt}
	
	 \section{Continuous coins}\label{sec:continuous}

		 A continuous adjustable coin, or CA-coin, enables the user to smoothly and continuously adjust the desired success probability. As a possible example, assume  
	  {an algorithm   which gets a composite integer $n$ as an input, and repeatedly attempts to find a divisor of $n$.    $\mu(P)$ is the cost (e.g. running time)  of finding a divisor with probability $P$, and the penalty $E$ is the loss implied by failing to find a divisor. The reusable version assumes that the algorithm is randomized, and the one time version assumes that the composite number is selected at random from a given distribution.} 

In general, a CA-coin  is  defined by a 
non-decreasing cost function $\mu(P):[P_{min},P_{max}] \rightarrow {\cal R} $ where $0<P_{min}<P_{max}\leq 1$ are the minimum and maximum success probabilities supported by the coin. The (optimal) cost of a   CA-coin  for a penalty $E$ is naturally defined as $$COST(\mu,E)=\min_{P\in[P_{min},P_{max}]}COST(P,\mu(P))  =\min_{P\in[P_{min},P_{max}1]}\mu(P) +(1-P)E.$$

	 We assume, as in the case of discrete A-coins,   that each CA-coin  is efficient, that is, for each $P_0\in [P_{min},P_{max}]$ there is a penalty $E_{P_0}$  s.t. $COST((P ,\mu(P )),E_{P_0})$ has a unique minimum at $P=P_0$. By arguments similar to the ones used in \cref{lem:convex},  this implies that a CA-coin $\mu$ is a nonnegative convex function on $[P_{min},P_{max}]$, and that $r(P)=\mu(P)/P$ is a strictly increasing function of $P$. We restrict our attention to \emph{regular} CA-coins 	$\mu(P)$ which are defined and twice continuously differentiable on the  segment $[P_{min},P_{max}]$. This implies that
		 	$$
	 	\frac{d \mu}{dP} > 0 ~~ \mbox{and} ~~ \frac{d^2\mu}{dP^2} > 0, ~~ P \in (P_{min},P_{max}) \, .
	 	$$
 \subsection{Single CA-coin, single toss} \label{sec:SingleCAcoinSingleToss}
	 Suppose we are given a single (efficient) CA-coin and are allowed to  toss  it once at most. That is, given $E$ we must decide whether we wish to use our CA-coin at all, and if we do, which value of $P$ we should
choose so as to minimize the expected cost. Recall that the expected cost is given by
	 $$
	 COST\left( \mu(P), P, E \right) = \mu(P) + (1-P)E = P\cdot r(P) + (1-P)E \, ,
	 $$
	
	 \noindent
	 where $r(P) = \mu / P$. Differentiating $COST$ with respect to $P$ we obtain
	 $$
	 \frac{d COST}{d P} = \frac{d \mu}{dP} - E \, .
	 $$
	
	 \noindent
	Differentiating $r(P)$ and noting that it is strictly increasing yield
	 $$
	 \frac{dr}{dP} = \frac{1}{P} \left( \frac{d \mu}{dP} - r \right) > 0 \, .
	 $$
	
	 \noindent
	 Observe that for $P_0\in(P_{min},P_{max})$, for each coin  $c_0=(P_0,\mu(P_0))$ there is a unique supporting line, namely, the tangent to $\mu$ at $c_0$.
	 This yields the following conclusions. Let $r_{min}=r(P_{min})$,  $E_{low}$ be the derivative of $\mu$ at $P_{min}$, and $E_{high}$ be the derivative of $\mu$ at $P_{max}$---see \cref{fig:CA-coin}. Then
	 \begin{enumerate}
	 	\item
	 	 	If $E\leq r_{min}$ then the coin is not beneficial for $E$, i.e., $COST(\mu,E)=E$. Otherwise the coin is beneficial for $E$, and the maximum benefit for a given penalty $E$ is attained as follows:
	 	\item
	 		If  $r_{min}< E\leq E_{low}$  then the benefit is maximized at $P_{min}$.
	 		 \item
	 		 	If  $E_{low}< E<E_{high} $ then the maximal benefit is attained at  {the unique} internal value   $P_{opt}\in (P_{min},P_{max})$,  {where $\frac{d \mu}{d P} = E$.}
	 		 	\item
	 		 	If  $E_{high} \leq E $  then   the benefit is maximized at $P_{max}$.
	   	 \end{enumerate}
   	 \begin{figure}[h!]
   	 	\begin{center}
   	 		  \vspace{-150pt}
   	 		\includegraphics[width= \columnwidth]{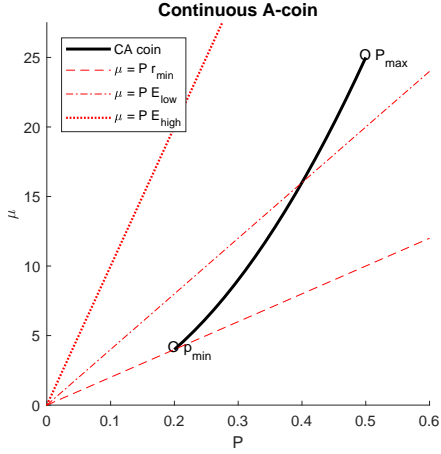}
   	 		 \vspace{-150pt}
   	 		\caption{
   	 			{$c_{min}=(P_{min},\mu(P_{min}))$, $c_{max}=(P_{max},\mu(P_{max}))$, and $r_{min}=\frac{\mu(P_{min})}{P_{min}}$. If $E<r_{min}$ then $c_{min}$ lies above the line $\mu=PE$, meaning that the CA-coin is not beneficial for $E$. For $E\in (r_{min},E_{low}]$ any supporting line of $\mu=PE$ passes through $c_{min}$. For  $E\in (E_{low},E_{high})$ there is a unique supporting line that passes through an internal point of the coin, and for  $E\in [E_{high},\infty)$ any supporting line passes through $c_{max}$.
   	 		} 	 		}
   	 		\label{fig:CA-coin}
   	 	\end{center}
   	 \end{figure}
	
	  \subsection{Reusable CA-coins}
	
	 Suppose next that we are given a single regular CA-coin which we may toss multiple times, selecting $P$ for each toss. The optimal strategy for minimizing the expected cost as a function of the number of tosses we are allowed, is obtained recursively from the observations of the previous subsection \ref{sec:SingleCAcoinSingleToss}. If the number of tosses is unlimited and $r_{min} < E$, then we can reduce the expected cost   to $r_{min}$   by repeatedly tossing the coin with $P = P_{min}$ until we succeed. If the number of tosses is bounded  then we can use the recursive approach described in \cref{sec:ReusableSimpleCoins}.
	
	  { When we are given a few regular CA-coins $\{\mu_1,\ldots,\mu_k\}$ which we may toss multiple times, we adapt a variant of strategies   used  in \cref{sec:ReusableSimpleCoins,sec:norep}: 
	 	 whenever  we need to toss a coin for a given penalty $E$,  we select $\mu_j$ for which $COST(\mu_j,E)$ is minimized. This can be done by computing $COST(\mu_i,E)$ for all $i\in [1,k]$.}

\section{Conclusion}
 {We present the concept of adjustable coins, which aims to model a scenario in which  various algorithms for solving a given problem can be applied: Each such algorithm is modeled by an adjustable coin, which is characterized by a cost and a probability of success. {The related optimization problem is: Given a set of independent A-coins and   a penalty for failing to solve the problem, find a  sequence of coin-tosses which minimizes the expected cost, subject to possible further restrictions.}
 	
  We note that all our solutions are by offline algorithms, which require that the full set (or sequence) of coins is given before the first coin is tossed. An interesting problem is   what conditions enable useful algorithms which use only partial information on the available coins (e.g., that the coins are drawn from a known distribution, or that only a limited number of coins is known ahead of time)}.

 \bibliographystyle{plain}
 \bibliography{a-coins}
\end{document}